\documentclass[a4paper,UKenglish,cleveref, autoref, thm-restate]{lipics-v2021}

\title{A Reflection Principle for Potential Infinite Models of Type Theory} 

\titlerunning{Potential Infinite Models of Type Theory} 

\author{Matthias Eberl}{LMU Munich, Theresienstr. 39, 80333 Munich, Germany}{matthias.eberl@mail.de}{https://orcid.org/0000-0002-2410-3747}{}

\authorrunning{Matthias Eberl} 

\Copyright{Matthias Eberl} 



\relatedversion{} 

\nolinenumbers 

\usepackage{mathpartir}

\def\EEE{\mathcal{E}}

\def\HHH{\mathcal{H}}
\def\III{\mathcal{I}}
\def\JJJ{\mathcal{J}}

\def\LLL{\mathcal{L}}
\def\MMM{\mathcal{M}}
\def\NNN{\mathcal{N}}

\def\epsilon{\varepsilon}
\def\phi{\varphi}
\def\rho{\varrho}
\def\theta{\vartheta}

\def\nat{\mathbb{N}}

\def\bool{\mathbb{B}}
\def\iffdef{:\!\iff}
\def\imp{\Rightarrow}
\def\prop{bool}

\def\pot#1{\mathcal{P}(#1)}

\def\nil{()}

\def\pmap{\stackrel{p}{\mapsto}}
\def\Pmap{\stackrel{p}{\longmapsto}}

\def\Qmap{\stackrel{q}{\longmapsto}}
\def\comp{\asymp}

\def\emb#1#2{emb_{#1}^{#2}}
\def\proj#1#2{proj_{#2}^{#1}}

\def\elem{\EEE\LLL}

\def\fil{\mathfrak{D}}

\def\invlim{\varprojlim}
\def\dirlim{\varinjlim}
\def\limf{\underline{\textrm{lim}}}
\def\y#1{\mathsf{#1}}

\def\up{\,\uparrow\!}

\def\type{{Typ}}
\def\y#1{\mathsf{#1}}
\def\yx#1{\mathsf{x_{#1}}}
\def\yxt#1#2{\mathsf{x}_\mathsf{#1}^{#2}}

\def\vecb#1{\boldsymbol{#1}}
\def\val#1{\lbrack\!\lbrack#1\rbrack\!\rbrack}
\def\tdt#1#2#3{#2 \vdash #1 : #3}
\def\tdts#1#2#3{#2 \mid #1 : #3}

\EventEditors{Delia Kesner, Eduardo Hermo Reyes, and Benno van den Berg}
\EventNoEds{3}
\EventLongTitle{29th International Conference on Types for Proofs and Programs (TYPES 2023)}
\EventShortTitle{TYPES 2023}
\EventAcronym{TYPES}
\EventYear{2023}
\EventDate{June 12-16, 2023}
\EventLocation{ETSInf, Universitat Polit\`{e}cnica de Val\`{e}ncia, Spain}
\EventLogo{}
\SeriesVolume{303}
\ArticleNo{6}
\begin{document}

\maketitle

\begin{abstract}
Denotational models of type theory, such as set-theoretic, domain-theoretic, or category-theoretic models use (actual) infinite sets of objects in one way or another. The potential infinite, seen as an extensible finite, requires a dynamic understanding of the infinite sets of objects. It follows that the type $nat$ cannot be interpreted as a set of all natural numbers, $\val{nat} = \nat$, but as an increasing family of finite sets $\nat_i = \{0, \dots, i-1\}$. Any reference to $\val{nat}$, either by the formal syntax or by meta-level concepts, must be a reference to a (sufficiently large) set $\nat_i$.

We present the basic concepts for interpreting a fragment of the simply typed $\lambda$-calculus within such a dynamic model. A type $\rho$ is thereby interpreted as a process, which is formally a factor system together with a limit of it. A factor system is very similar to a direct or an inverse system, and its limit is also defined by a universal property. It is crucial to recognize that a limit is not necessarily an unreachable end beyond the process. Rather, it can be regarded as an intermediate state within the factor system, which can still be extended.

The logical type $\prop$ plays an important role, which we interpret classically as the set $\{true, false\}$. We provide an interpretation of simply typed $\lambda$-terms in these factor systems and limits. The main result is a reflection principle, which states that an element in the limit has a ``full representative'' at a sufficiently large stage within the factor system. For propositions, that is, terms of type $\prop$, this implies that statements about the limit are true if and only if they are true at that sufficiently large stage.
\end{abstract}

\section{Introduction}

In set theory, infinite sets are given by the dictum of the axiom of infinity. There is no idea of ``construction'' or ``approximation'' involved --- it is a \emph{static} concept where only existence is required, without any way to get to these sets. In contrast, consider the constructions of infinite sets as limits of direct and inverse systems. These sets are approximated and so can be understood from the perspective of a potential infinite. Moreover, they possess the structure of their approximating parts. Finally, and most importantly from a finitistic perspective, there is no necessity to ``jump over'' to an absolute, actual infinite limit set if all states of the system are finite. Instead, if one takes care of all the stages of one's investigation, then a sufficiently large state within the system is sufficient and is a full substitute for an infinite limit set.

From a consequent potentialist's point of view, it is actually a misuse of language to call a set infinite. Since the potentialist's view is a form of finitism, and since sets are given by their extension, every set is finite in this regard. Therefore, what can be considered infinite is the type, say $nat$, whereas the set $\nat$ would be more accurately described as \emph{indefinitely extensible}. Thus, the usual terminology, saying that a set is infinite, means that its type is infinite and that the extension of the type, i.e., the set of elements of that type, is given by some \emph{indefinitely large} state (or \emph{sufficiently large} state) in an indefinitely extensible system of finite sets.

A system, introduced in Section \ref{factsec}, formalizes a changing totality of objects, an ``open'' and potentially unending process in which both the collection and its elements expand simultaneously. This process allows a (relative) completion or compactification by constructing a limit, which temporarily ends or ``closes'' the process.

The concept of partiality becomes relevant in extensible systems. It is a consequence of the fact that objects can be generated, or if you prefer a less constructive language, detected, so that they do not exist from the beginning. Partiality first appears in direct systems, where objects do not exist at the initial stages. In this situation, partiality is not difficult to deal with, but it becomes demanding in the function space construction, especially in the presence of higher types.

A few words about the notation. The term ``iff'' is an abbreviation for ``if and only if''. $\nat$ refers to the set of natural numbers $\{0,1,2,\dots\}$, $\nat^+ := \nat \setminus \{0\}$ and $\nat_i := \{0, \dots, i-1\}$. We write $[A \to B]$ for the function space of all functions with domain $A$ and codomain $B$, and $\pot{\MMM}$ for the power set of $\MMM$.

\subsection{Extensibility, Coinduction and Domain Theory}
\label{extsec}

A fundamental concept in the context of the potential infinite is \emph{extensibility}. The main modes of extensibility are the creation of new objects and the creation of new knowledge about existing objects. However, we understand the latter as a differentiation of an object or an identification of several objects. In the case of differentiation, there may be multiple versions of an object at a later stage. When considering limit constructions, this is different from the understanding of the accumulation of information. To illustrate this point, let us consider inductive and coinductive definitions. These are related to adding and differentiating objects respectively, but they are not the only way.

\begin{enumerate}
\item The dynamic reading of an inductive definition leads to an infinite process of creating objects, related to direct systems.
\item The dynamic reading of a coinductive definition leads to an infinite process of differentiating objects, related to inverse systems.
\end{enumerate}

To give an example, let $Seq$ denote finite 0-1-sequences and $Seq^\infty$ stand for infinite 0-1-sequences. Then the algebra with constructors $nil : \{\ast\} \to Seq$ and $append : Seq \times \nat_2 \to Seq$ inductively defines the structure of $Seq$, whereby $\nat_2 = \{0,1\}$. The coinductive definition of $Seq^\infty$ has one destructor-pair $(head,tail) : Seq^\infty \to \nat_2 \times Seq^\infty$ and defines a coalgebra. To give these structural definitions a dynamic reading requires an index set, which will be $\nat^+$ and $\nat$. 

In case of inductive definitions one starts with the element $\ast$. The states of $Seq$ are thus $Seq_0 = \{\ast\}$, $Seq_1 = \{\nil\}$, $Seq_2 = \{\nil, 0, 1\}$, $Seq_3 = \{\nil, 0, 1, 00, 01, 10, 11\}$ and so on. The inductive definition of $Seq$ gives rise to a direct system $(Seq_i)_{i \in \nat^+}$ with subset inclusion as embedding. A direct system is more general than this construction by an inductive definition, as it allows for the possibility of non-injective embeddings. This corresponds to the addition and identification of objects in a single process and is relevant for a construction of quotients.

For a dynamic understanding of coinductive definitions, start with a ``generic'' element, say $s$, so $Seq^\infty_0 = \{s\}$. The destructor-pair $(head,tail)$ gives $Seq^\infty_1 = \{0s, 1s\}$, $Seq^\infty_2 = \{00s, 01s, 10s, 11s\}$, and so on. The coinductive definition of $Seq^\infty$ allows one to increase the knowledge about the sequence, which we understand as a process of differentiation. The projections $\proj{i'}{i} : Seq^\infty_{i'} \to Seq^\infty_{i}$ then remove this information gained by differentiation, e.g., $\proj{2}{1}(01s) = 0s$. We read $01s \in Seq^\infty_2$ and $0s \in Seq^\infty_1$ as states of the same (infinite) object, differing only in the amount of information that we have about that object. The inverse system $(Seq^\infty,proj)$ corresponds to subsets $A \subseteq \nat$ if we take $A = \{i \in \nat \mid s_i = 1\}$ for $(s_i)_{i \in \nat} \in Seq^\infty$. 

This idea of differentiating objects is also related to Brouwer's concept of (lawlike and lawless) choice sequences within a spread --- the standard text on this subject is \cite{troelstra1979}. However, our approach does not involve the constructive reasoning that is inherent in the definition of a spread.

Differentiating infinite objects differs from domain theory \cite{abramsky1994domain}, where the ideal, infinite elements are ideals (being infinite sets) of their approximating parts, which are the compact elements. In domain theory, each approximation is seen as a different object, not as states of one object, and the ideal completion uses the idea that sets are \emph{actual} infinite. In other words, the idea of differentiating infinite objects, as formulated in coinductive definitions and inverse systems, is not present in domain theory. Domain theory uses inductive definitions and least fixed point constructions.

Another way to get these ideal elements in domains is to think of an infinite domain as a bilimit of finite domains. A bilimit is a limit construction where direct and inverse limits coincide. One could interpret this coincidence as a reduction of inverse systems to direct systems. In the approach here, the limit construction is, roughly speaking, between the direct and the inverse limit, since the predecessor relation $\pmap$ (which will be introduced soon) is, again roughly speaking, between the embedding $emb$ and the projection $proj$. In Section \ref{factsec} we will introduce such an embedding-projection pair, associated with the predecessor relation $\pmap$. Direct and inverse systems can then be seen as extreme situations of adding and differentiating objects, while the system based on the predecessor relation $\pmap$ is in general a combination of both.

\subsection{Formalizing the Potential Infinite}
\label{forminf}

A potential infinite set $\MMM$ is a dynamic, finite set. To formalize this consequently, any reference to $\MMM$ can only be made by reference to some finite state $\MMM_i$. A completed totality of all elements or all states of this set does not exist. For example, there exist states $\nat_i$ of the set of natural numbers, but there is neither a complete set $\nat$ nor the complete family of all sets $\nat_i$. The latter has to do with the fact that, for a consequent reading, this finitistic view applies to meta-level concepts as well, in particular to the index set from which the indices $i$ are taken; see Section \ref{metasec} for more on this.

In a first step we introduce indices $i$, given by a directed index set $\III$, i.e., a set $\III$ of stages together with a binary, reflexive and transitive relation $\leq$ so that any finite set of indices has an upper bound. A potential infinite set is thus a family $\MMM_\III$ of states $\MMM_i$, $i \in \III$, where all sets are finite.

We want to express that two elements $a_i \in \MMM_i$ and $a_{i'} \in \MMM_{i'}$ at different stages $i$ and $i'$ are equal. This equality is not given by an equivalence relation, but by a family of reflexive relations $\pmap_{i',i} \, \subseteq \MMM_{i'} \times \MMM_{i}$ for $i' \geq i$. We write $a_{i'} \pmap a_i$ for $(a_{i'}, a_{i}) \in \ \pmap_{i',i}$, saying that $a_i$ is a predecessor of $a_{i'}$, and we use $\pmap$ as an abbreviation for $\pmap_{i',i}$. Reflexivity means that $a_i \pmap a_i$ for all $i \in \III$ and $a_i \in \MMM_i$. We do not require $\pmap$ to be transitive.\footnote{The reason is this: The index set of the function space is $\III \times \JJJ$ with pairs written as $i \to j$ (c.f.~Section \ref{funspacefac}). Call an extension of the index from $i \to j$ to $i \to j'$, $j' \geq j$, \emph{covariant} and from $i \to j$ to $i' \to j$, $i' \geq i$, \emph{contravariant}. Then a combination of two different kind of extensions may fail to be a correct extension. For an example, see \cite{eberl2023}.} A \emph{system} $(\MMM_\III,\pmap)$ consists of a family $\MMM_\III$ together with reflexive relations $\pmap$. Later we will introduce further properties in order to allow a function space construction, leading to the notion of a \emph{factor system}
\cite{eberl2023}.

Let types $\rho \in \type$ be given. In this paper they consist of some base types $\iota$, including type $\prop$ for propositions, which are interpreted classically as the Boolean values $true$ and $false$, and a type constructor $\to$. So types are $\rho ::= \iota \,\,|\,\, \rho \to \rho$. A (typing) context $\Gamma$ is a list of types $(\rho_0, \dots, \rho_{n-1})$. More explicitly we can write $(\yx{0} : \rho_0, \dots, \yx{n-1} : \rho_{n-1})$ for the context $\Gamma$, since we use a fixed list of variables $\yx{0},\yx{1},\yx{2}, \dots$ in Section \ref{sysemlam}. The empty context is $\nil$ and $\Gamma.\rho$ denotes the context $\Gamma$, extended by $\rho$. We assume that for each type $\rho$ there is an index set $(\III_\rho,\leq)$, so each type comes with its own set of stages. For instance, the index set for type $nat$ is $(\nat^+,\leq)$ with $\MMM_i = \nat_i$ for $i \in \nat^+$.\footnote{The use of $\nat^+$ instead of $\nat$ as index set is done for technical reasons: A factor system has projections between different states $\nat_i$, and there is no projection from $\nat_i$, $i > 0$ into $\nat_0 = \emptyset$.} Another example is type $\prop$ with the singleton set $(\{\prop\},=)$ of one index $\prop$ and $\MMM_{\prop} := \bool := \{true,false\}$. For a typing context $\Gamma = (\rho_0, \dots, \rho_{n-1})$ define $\III_\Gamma := \III_{\rho_0} \times \dots \times \III_{\rho_{n-1}}$ and endow  $\III_\Gamma$ with the product order. 

Since we want to consider sets, relations and functions as objects, i.e., as elements of further potential infinite sets, we have to consider elements as dynamic entities as well. A dynamic object $a$ of type $\rho$ is given by its states $a_i$, $i \in \III_\rho$. As with sets, any reference to $a$ can only be a reference to one of its states $a_i$.

A dynamic object needs not be defined on all stages $i \in \III_\rho$, so the index set $\III_a$, for which a state $a_i$ of $a$ exists, is a subset of $\III_\rho$, i.e., $\III_a \subseteq \III_\rho$. One of the basic requirements is that every object has ``sufficiently many'' indices $\III_a$. For natural numbers, or more generally for base type objects, the situation is simple: If a number $n$ occurs at stage $n+1$, i.e., $n \in \nat_{n+1}$, it will be there for all future stages $m > n+1$, i.e., $n \in \nat_m$, so the index set $\III_n = \{n+1, n+2, \dots\}$ is an up-set. By \emph{up-set} we mean a non-empty, upward closed subset, in this case a subset of $\nat^+$. For higher-order functions, however, the situation is less straightforward and is one of the main challenges of this approach.

The concept of a potential infinite has two aspects, a cardinal aspect $\fil$ and an ordinal aspect $\ll$. Let $C = (i_0, \dots, i_{n-1}) \in \III_\Gamma$ be a list of indices $i_0 \in \III_{\rho_0}, \dots, i_{n-1} \in \III_{\rho_{n-1}}$, for $\Gamma = (\rho_0, \dots, \rho_{n-1})$, called \emph{state context}. Write $C.i$ for the extension of $C$ by $i$.
\begin{enumerate}
\item The cardinal aspect is given by a set $\fil_\Gamma \subseteq \pot{\III_\Gamma}$ for each context $\Gamma$. $\HHH \in \fil_\Gamma$ says that there are \emph{indefinitely many}, or \emph{sufficiently many} contexts in $\HHH$.

\item The ordinal aspect is given by relations $\ll^{\Gamma.\rho} \, \in \fil_{\Gamma.\rho}$. Let $C \ll i$ stand for $C.i \in \ll^{\Gamma.\rho}$, meaning that a stage $i$ is \emph{indefinitely large}, or \emph{sufficiently large} relative to the state context $C$.
\end{enumerate}

The sets $\fil_\Gamma$ satisfy the following properties: Each set $\HHH \in \fil_\Gamma$ is cofinal --- recall that cofinality of $\HHH$ means $\forall C \in \III_\Gamma \, \exists C' \geq C$ with $C' \in \HHH$. This minimal requirement simply states that we will always find a context in $\HHH$ beyond any bound. $\fil_\Gamma$ is closed under supersets, so ``more than indefinitely many is indefinitely many''. Furthermore, an up-set on $\III_\Gamma$ always has indefinitely many indices. Finally, the main restriction is that $\fil_\Gamma$ is closed under intersections. This is necessary in order to guarantee that a relation between two objects can be established on indefinitely many indices. This amounts to saying: 
\begin{equation}
\label{filtereq}
\tag{Filter}
\fil_\Gamma \text{ is a proper filter on } \{\HHH \subseteq \fil_\Gamma \mid \HHH \text{ is cofinal}\} \text{ and } \fil_\Gamma \text{ contains all up-sets}.
\end{equation}

We will use the locution \emph{$\fil$-many indices}, which refers to a set in $\fil_\Gamma$, and \emph{cofinal many indices}, which refers to a cofinal index set. The interpretation in a potential infinite structure, which we introduce in Section \ref{sysemlam}, will be relative to $\fil$ (and later also relative to $\ll$, when we introduce the universal quantifier in a subsequent paper). 

A basic theme of the potential infinite is \emph{dependency}. In particular, there are no fixed sets $\fil$ and $\ll$, but these are parameters that depend on factors of the concrete mathematical investigation and the state of it. Another way of expressing this is to say that $\fil$ and $\ll$ are intensional notions whose extension depends on the context of investigation. In the same way that $\fil$ and $\ll$ depend on the investigation, some concepts, in turn, depend on $\fil$ and $\ll$; the notion of continuity depends on $\fil$ and the interpretation of the universal quantifier depends on $\ll$. It is possible to define continuity on limit sets without reference to the underlying system. The basic relation is then a family of PERs and we will explore this structure in more detail in a separate paper. Moreover, in this paper we deal only with the cardinal aspect, i.e., $\ll$ is not considered here.

\subsection{Relation to Constructive Approaches}

We are only investigating the idea of a potential infinite, not that of constructivity, decidability, complexity or knowledge about existence, which are important concepts in intuitionism \cite{troelstra1988constructivism} and in theories about computability \cite{longley2015higher}. 
The common models of intuitionistic logic, such as Kripke models, or more generally topos-theoretic models \cite{maclane2012sheaves}, use unbounded universal quantification. For instance, in a Kripke model, the validity of a universal quantified formula uses a reference to all ``future'' nodes --- there could be infinitely many of them and at each such node the carrier set could be an infinite set as well. In our approach, only finite sets are used, and in a consequent finitistic view only finitely many of them.

We use classical logic. This is because it has a simpler model than constructive models and is more widely used. Furthermore, the results presented here rely on the fact that classical logic has a finite number of truth values. This is an important difference from intuitionistic logic, which has, if truth values are used at all, infinitely many of them.

The novel aspect of this potential infinite model is the introduction of \emph{state judgements}, which are refinements of typing judgements. If $\tdt{\y{r}}{\Gamma}{\rho}$ is such a typing judgement, then a state judgement has the form $\tdts{\y{r}}{C}{i}$, where $C$ is a stage of the context $\Gamma$ and $i$ is a stage of the type $\rho$. The idea that we can only refer to infinite objects via a specific state is reflected in the fact that the primary object of interpretations is a state judgement $\tdts{\y{r}}{C}{i}$. It is interpreted in a family of factor systems, which have only ``local'' application functions given as 
\[
App_{i,j} : [\MMM_i \to \NNN_j] \times \MMM_i \to \NNN_j, \ (f_{i \to j}, a_i) \mapsto f_{i \to j}(a_i).
\]

Additionally, the typing judgement $\tdt{\y{r}}{\Gamma}{\rho}$ has an interpretation in the limit set. However, the notion of a limit of a factor system depends on the notion of ``indefinitely many stages'', which we formalize as sets $\fil_\rho$. This is in particular relevant for the function space and the possibility to define a ``global'' application function $App : [\MMM \to \NNN] \times \MMM \to \NNN, (f, a) \mapsto f(a)$ on limit sets, leading to a common \emph{extensional type structure} \cite{barendregt_1984}. The question which higher-order functions exist and whether all local applications together yield a global application depends on the properties of $\fil_\rho$. A global application function is available for all first-order functions, however, this may not be the case for higher-order functions. In contrast, in a type structure, also referred to as \emph{(typed) applicative structure}, a global application is available for all types. This holds analogously for \emph{Kripke applicative structures} \cite{mitchell1991kripke}.

\subsection{The Meta-Theory}
\label{metasec}

The meta-theory in which the concepts are developed is classical higher-order logic, as formalized in Church's simple type theory \cite{church1940formulation}. This theory will also serve as the investigated theory. At the object-level, we develop a potential infinite model in order to interpret typed $\lambda$-terms. At the meta-level these two views of infinity are relevant:
\begin{enumerate}
\item One accepts actual infinity at the meta-level. In this case, a limit can be seen as the usual actual infinite set beyond the system. This view allows a comparison of an actual infinite model, given as limit structure, with the potential infinite part, i.e., the system.

\item One uses the view that infinity is an extensible finite. Then a limit is an intermediate state of the system. This is the consistent realization of the finitistic approach.
\end{enumerate}

The reflection principle, which is the main theorem of this paper, states in the first case that all objects in the limit, including propositions, have a counterpart in the factor system. So it says something like this: Whatever exists and holds under the assumption that actual infinite sets exist, already exists and holds at a sufficiently large stage in the system. For infinite objects, these can be seen as approximations. For type $\prop$, which is interpreted in classical logic as a finite set of truth values $\{true,false\}$, the values are the same in the actual infinite limit set and at a finite stage of the system. This is because these values are not approximated.

In the second case, the reflection principle is only a means to show the correctness of the interpretation. However, a consistent realization requires that an infinite set on meta-level, like the index set $\III$, is only available at a stage $j$. A consequent realization in type theory moreover uses a type in place of the index set, together with a term $\leq$, which is then shown to be reflexive, transitive, and directed.

The use of a classical meta-theory instead of a constructive one is not essential here. We could also take an intuitionistic type theory at meta-level and develop most of the model theory in a pure constructive way. So we expect that the model construction can be formalized in common proof assistants such as Coq, Lean or Agda. We need, however, a bit of classical reasoning, at least when introducing the universal quantifier. To prove the reflection principle with universal quantifier as an extension of Theorem \ref{mainthm}, one has to do a kind of L\"owenheim-Skolem construction. This requires that a universal quantified formula is either true (at all stages), or it is false at some stage and is false at all later stages, too. It is of greater significance that, for a consequent realization, in which only potential infinities are used at meta-level, the proof assistant must implement state judgments.

\subsection{Structure of the Paper}

We already started in Section \ref{forminf} to formalize the potential infinite as a dynamic concept, which replaces an actual infinite set with a \emph{factor system}, which is, roughly, a generalization of a direct and inverse system. The concept of factor systems was first introduced in \cite{eberl2023}. In Section \ref{factlimsec} we reiterate the definition of a factor system and add further definitions that are necessary for an interpretation. These are primarily the notions of a direct and inverse factor system, which are required to interpret variables. We show the construction of the function space between two factor systems and elucidate the notions of a target and a limit of a factor system. As with direct and inverse systems, limits are targets that satisfy a universal property and have a concrete construction. In Section \ref{moretlsec} we demonstrate how to introduce an application that makes limit sets a type structure.

Section \ref{sysemlam} introduces a judgement for states, parallel to judgements for types. These are defined on a fragment of the simply typed $\lambda$-calculus, which we call \emph{core fragment}. Based on these state judgements we give a first version of an interpretation of $\lambda$-terms in the core fragment, not including any constants. In particular, the present paper does not yet include logic. We give an interpretation of types and terms with two parts, one is within the system, the other is in its limit. Based on this interpretation we show a first version of a reflection principle, which says that an element $a$ of type $\rho$ in the limit set is reflected by an element $a_i$ at some stage $i \in \III_\rho$ in the system. This element $a_i$ is an approximation of $a$, and at the same time it fully represents $a$. If logic is included, then the representation includes all propositions about these elements, so anything we can say about $a$ is true if and only if it is true for $a_i$.

\section{Factor Systems and their Limits}
\label{factlimsec}

Factor systems and their limits have been introduced in \cite{eberl2023}. In this paper, we summarize their properties. The main concept is that of a factor system, those of a prefactor system (with embeddings/projections) are afferent notions. In addition to that, we introduce the notion of a \emph{stable} system, which is a natural notion to prove stronger properties, although these are not necessary here. We prove that stability is closed under both the function space construction and the limit construction.

Relevant for the interpretation in Section \ref{sysemlam} are the specific forms of a \emph{direct} and \emph{inverse} factor system. A direct factor system is similar to a direct system, it is more specific in the sense that the embedding is part of an embedding-projection pair. On the other hand it is more general in the sense that equations hold only up to an equivalence relation. The same holds for inverse factor systems and inverse systems. Moreover, we introduce the concept of a \emph{homomorphism} between two systems. 

The subsequent Lemmata \ref{elemlem} and \ref{elemlem2} are extended versions of corresponding lemmata in \cite{eberl2023}, which use the property (\ref{filtereq}), not only cofinality. These versions are necessary to prove Corollary \ref{limitfuncor}, which states that the function space construction commutes with the limit construction. This is a prerequisite for the definition of a model. Proposition \ref{funspprop} and Corollary \ref{invdiecor} describe how direct and inverse factor systems extend to the function space and to the limit. Both are essential to prove the reflection principle for variables.

\subsection{Factor Systems}
\label{factsec}

$\III$ will always denote a non-empty directed index set with preorder $\leq$. A \emph{system} is a pair $(\MMM_\III,\pmap)$ consisting of a family $\MMM_\III := (\MMM_i)_{i \in \III}$ and reflexive (for $i = i'$) relations $\pmap$ on $\MMM_{i'} \times \MMM_i$ for $i' \geq i$. Two elements $a_i \in \MMM_{i}$ and $b_j \in \MMM_{j}$ are \emph{consistent}, written as $a_i \comp b_j$, iff there is an index $i' \geq i, j$ and some $a_{i'} \in \MMM_{i'}$ such that $a_{i'} \pmap a_i$ and $a_{i'} \pmap b_j$. As a convention, whenever we use a suffix $i \in \III$ for some element, this refers to the state, e.g.~$a_i \in \MMM_i$. An important special case is that the relations $\pmap$ are partial functions, which is equivalent to:
\begin{equation}
\tag{Fun}
\label{idmapeq}
a_i \comp b_i \iff a_i \pmap b_i \iff a_i = b_i
\end{equation}
for all $a_i, b_i \in \MMM_i$ and all $i \in \III$. A system that satisfies (\ref{idmapeq}) is called \emph{standard}. $(\MMM_\III,\pmap)$ is a \emph{prefactor system} iff it is a system satisfying
\begin{equation}
\label{pmapeq} 
\tag{Factor} 
a_{i'} \comp a_i \iff a_{i'} \pmap a_i
\end{equation}
for all $a_i \in \MMM_{i}$ and $a_{i'} \in \MMM_{i'}$ with $i \leq i'$. The relation $\comp$ is then an equivalence relation on a single set $\MMM_i$ with $a_i \comp b_i \iff a_i \pmap b_i \iff b_i \pmap a_i \ \text{ for } a_i, b_i \in \MMM_i$. In a prefactor system $b_{i'} \comp a_{i'} \pmap a_i$ implies $b_{i'} \pmap a_i$, but sometimes we want to have the ``dual'' property as well, which we call \emph{stability}. Although we can do without stability for most properties, it is an obvious requirement, and all natural examples satisfy this property. Note that a system that satisfies (\ref{idmapeq}) is automatically stable.

\begin{definition}
A system $(\MMM_\III,\pmap)$, and in particular the relation $\pmap$, is called \emph{stable} iff for all $i' \geq i$, all $a_{i'} \in \MMM_{i'}$ and $a_i, b_i \in \MMM_i$
\begin{equation}
\label{stabeq} 
\tag{Stab} 
a_{i'} \pmap a_i \comp b_i \ \imp \ a_{i'} \pmap b_i.
\end{equation}
\end{definition}

A family $emb = (\emb{i}{i'})_{i \leq i'}$ of $\comp$-embeddings consists of $\comp$-preserving maps $\emb{i}{i'} : \MMM_i \to \MMM_{i'}$ satisfying $\emb{i}{i}(a_i) \comp a_i$ and $\emb{i'}{i''}(\emb{i}{i'}(a_i)) \comp \emb{i}{i''}(a_i)$. The requirement $\comp$-preserving means that $\emb{i}{i'}(a_i) \comp \emb{i}{i'}(b_i)$ holds for $a_i \comp b_i$, for all $a_i, b_i \in \MMM_i$. Similar to a family of $\comp$-embeddings, $\comp$-projections $proj = (\proj{i'}{i})_{i \leq i'}$ consist of $\comp$-preserving maps $\proj{i'}{i} : \MMM_{i'} \to \MMM_i$ satisfying $\proj{i}{i}(a_i) \comp a_i$ and $\proj{i'}{i}(\proj{i''}{i'}(a_{i''})) \comp \proj{i''}{i}(a_{i''})$. Moreover, $\comp$-embeddings $emb$ together with $\comp$-projections $proj$ form an \emph{$\comp$-embedding-projection pair} iff $\proj{i'}{i}(\emb{i}{i'}(a_i)) \comp a_i$ holds for all $a_i \in \MMM_i$ and all $i \leq i'$. 

The $\comp$-embeddings $emb$ and $\comp$-projections $proj$ are \emph{coherent} if they satisfy for all indices $i \leq i' \leq i''$
\begin{align}
\label{embcond}
\tag{Emb}
a_{i'} \pmap a_i \ &\imp\ \emb{i'}{i''}(a_{i'}) \pmap a_i \text{ and} \\
\tag{Proj}
\label{projcond}
a_{i''} \pmap a_i \ &\imp\ \proj{i''}{i'}(a_{i''}) \pmap a_i \text{ resp.}
\end{align}

A $\comp$-embedding-projection pair $(emb,proj)$ is coherent if $emb$ and $proj$ are both coherent. Property (\ref{embcond}) implies that $\emb{i}{i'}(a_i) \pmap a_{i}$ holds for all $a_i \in \MMM_i$, and in case that Property (\ref{idmapeq}) holds, $\pmap$ is a \emph{partial surjection}.

\begin{definition}
A \emph{prefactor system with embeddings} is a prefactor system $(\MMM_\III,\pmap)$ with coherent $\comp$-embeddings $emb$. A \emph{prefactor system with projections} is a prefactor system $(\MMM_\III,\pmap)$ with coherent $\comp$-projections $proj$.  A \emph{factor system} is a prefactor system $(\MMM_\III,\pmap)$ with a coherent $\comp$-embedding-projection pair $(emb,proj)$. A prefactor system is \emph{direct} iff it has coherent $\comp$-embeddings which satisfy 
\begin{equation}
\label{dircond}
\tag{Dir}
a_{i'} \pmap a_i \iff a_{i'} \comp \emb{i}{i'}(a_i).
\end{equation}

A prefactor system is \emph{inverse} iff it has coherent $\comp$-projections which satisfy 
\begin{equation}
\label{invcond}
\tag{Inv}
a_{i'} \pmap a_i \iff \proj{i'}{i}(a_{i'}) \comp a_i.
\end{equation}
\end{definition}

If $\pmap$, $emb$ and $proj$ are known, we often call $\MMM_\III$ a factor system (and likewise with prefactor systems with embeddings/projections).

\begin{example}
\label{dirinvex}
The embeddings $emb$ in a direct system $(\MMM_\III, emb)$, with $a_{i'} \pmap a_i \iffdef \emb{i}{i''}(a_{i}) = \emb{i'}{i''}(a_{i'})$ for some $i'' \geq i, i'$, are automatically $\comp$-embeddings, since they preserve $\comp$. If there are $\comp$-projections $proj$, such that $emb$ and $proj$ form a $\comp$-embedding-projection pair, then $(\MMM_\III,\pmap,emb,proj)$ is a direct factor system.

An inverse system $(\MMM_\III, proj)$ satisfies (\ref{idmapeq}), so the projections $proj$ are automatically $\comp$-projections. If there are $\comp$-embeddings $emb$, such that $emb$ and $proj$ form a $\comp$-embedding-projection pair, then $(\MMM_\III,\pmap,emb,proj)$, with $a_{i'} \pmap a_i \iffdef \proj{i'}{i}(a_{i'}) = a_i$, is an inverse factor system.
\end{example}

\begin{example}
\label{natex}
Consider $(\nat_i)_{i \in \nat^+}$ with the embedding-projection pair $\emb{i}{i'} : \nat_i \to \nat_{i'}$, $n \mapsto n$ and $\proj{i'}{i} : \nat_{i'} \to \nat_i$ with $n \mapsto \min(n,i-1)$. In all three cases for $\pmap$ they form an $\comp$-embedding-projection pair and $((\nat_i)_{i \in \nat^+},\pmap,emb,proj)$ is a factor system:
\begin{enumerate}
\item The standard model of the natural numbers has $\nat_{i'} \ni n \pmap n \in \nat_i$ for all $n < i$, which is a direct factor system.
\item The definition $n' \pmap n \iffdef \proj{i'}{i}(n') = n$ makes it an inverse factor system.
\item With $n' \pmap n$ for all $n' \in \nat_{i'}$ and $n \in \nat_i$, the factor system $(\nat_i)_{i \in \nat^+}$ is direct and inverse.
\end{enumerate}

The second example basically adds an infinite number to $\nat$, while the last example is artificial, but shows the difference to direct an inverse limits when we pick up this factor system again in Example \ref{dirinvex2}. One can easily check that $(\nat_i)_{i \in \nat^+}$ is stable in all three cases, and in the first two cases Property (\ref{idmapeq}) is also satisfied.
\end{example}

\begin{lemma}
\label{dirinvlem}
Let $\MMM_\III$ be a prefactor system and given indices $i'' \geq i' \geq i$.
\begin{enumerate}
\item If $\MMM_\III$ is direct, then $a_{i''} \pmap a_i$ implies $a_{i''} \pmap \emb{i}{i'}(a_i)$.
\item If $\MMM_\III$ is inverse, then $a_{i''} \pmap a_{i'}$ implies $a_{i''} \pmap \proj{i'}{i}(a_{i'})$.
\end{enumerate}
\end{lemma}

\begin{proof}
First, $a_{i''} \pmap a_i$ implies $\emb{i'}{i''}(\emb{i}{i'}(a_i)) \comp \emb{i}{i''}(a_i) \comp a_{i''}$ by (\ref{dircond}), hence $\emb{i'}{i''}(\emb{i}{i'}(a_i)) \pmap a_{i''}$. From (\ref{embcond}) we deduce $\emb{i'}{i''}(\emb{i}{i'}(a_i)) \pmap \emb{i}{i'}(a_i)$ and thus $a_{i''} \pmap \emb{i}{i'}(a_i)$ by (\ref{pmapeq}). For the second clause let $a_{i''} \pmap a_{i'}$, then $\proj{i''}{i'}(a_{i''}) \comp a_{i'}$ by (\ref{invcond}), hence $\proj{i''}{i}(a_{i''}) \comp \proj{i'}{i}(\proj{i''}{i'}(a_{i''})) \comp \proj{i'}{i}(a_{i'})$. Consequently $a_{i''} \pmap \proj{i'}{i}(a_{i'})$ by (\ref{invcond}) again, as claimed.
\end{proof}

Compare these properties with (\ref{embcond}) and (\ref{projcond}), which hold in any prefactor system with embeddings/projections.

\begin{definition}
\label{homdef}
A \emph{homomorphism} $\Phi = (\Phi_0, (\Phi^i)_{i \in \III})$ between two systems $(\MMM_\III, \pmap)$ and $(\NNN_\JJJ, \pmap)$ consists of maps $\Phi_0 : \III \to \JJJ$ and $\Phi^i : \MMM_i \to \NNN_{\Phi_0(i)}$ such that $\Phi_0$ is monotone and for all $i \leq i'$
\begin{equation}
\label{homeq}
a_{i'} \pmap a_i \ \imp \ \Phi^{i'}(a_{i'}) \pmap \Phi^i(a_i).
\end{equation}

If (\ref{homeq}) is an equivalence, then $\Phi$ is said to be \emph{strong}. A homomorphism between two prefactor systems with embeddings $(\MMM_\III, \pmap, emb)$ additionally satisfies $\Phi^{i'}(\emb{i}{i'}(a_i)) = \emb{j}{j'}(\Phi^i(a_i))$, $j := \Phi_0(i)$, $j' := \Phi_0(i')$, a homomorphism between two prefactor systems with projections $(\MMM_\III, \pmap, proj)$ satisfies $\Phi^{i}(\proj{i'}{i}(a_{i'})) = \proj{j'}{j}(\Phi^{i'}(a_{i'}))$, and a homomorphism between two factor systems satisfies both equations.

We call $\Phi$ injective (surjective) iff every map in $\Phi$ is injective (surjective resp.). An \emph{isomorphism} $\Phi$ between two systems (prefactor systems with embeddings/projections, factor systems) is a homomorphism with a further homomorphism $\Psi$ as its inverse, i.e., each part of $\Psi$ is inverse to that of $\Phi$. So an isomorphism is automatically strong. We write $\MMM_\III \simeq \NNN_\JJJ$ if an isomorphism between $\MMM_\III$ and $\NNN_\JJJ$ exists.
\end{definition}

In the following we will use homomorphisms $\Phi$ only for the special situation that $\III = \JJJ$ and $\Phi_0$ is the identity map. In that case we do not mention $\Phi_0$.

\subsection{The Function Space}
\label{funspacefac}

The function space of two factor systems $\MMM_\III$ and $\NNN_\JJJ$, denoted as $[\MMM_\III \to \NNN_\JJJ]$, is a family of (finite) sets $[\MMM_i \to \NNN_j]$ indexed by pairs $(i,j) \in \III \times \JJJ$ with product order, whereby we write $i \to j$ for such an index in $\III \times \JJJ$. The set $[\MMM_i \to \NNN_j]$ consists of all (total) functions $f : \MMM_i \to \NNN_j$ which preserve $\comp$, i.e., which satisfy $f(a_i) \comp f(b_i)$ for $a_i \comp b_i$. If the relations $\pmap$ are partial functions on $\MMM_\III$ and $\NNN_\JJJ$, then $\comp$ is the identity on $\MMM_i$ and $\NNN_j$ and $[\MMM_i \to \NNN_j]$ simply consists of all functions from $\MMM_i$ to $\NNN_j$.

Let $f \in [\MMM_{i} \to \NNN_{j}]$, $f' \in [\MMM_{i'} \to \NNN_{j'}]$ and $i \to j \leq i' \to j'$, i.e., $i \leq i'$ and $j \leq j'$. The basic relation $\pmap$ on the function space is a logical relation \cite{plotkin1973lambda}. It is thus defined as
\begin{equation*}
f' \pmap f \ \iffdef \ a_{i'} \pmap a_i \text{ implies } f'(a_{i'}) \pmap f(a_i)
\end{equation*}
for all $a_{i'} \in \MMM_{i'}$ and $a_i \in \MMM_{i}$. The embedding-projection pair for the function space is defined in the usual way:
\begin{align*}
\emb{i \to j}{i' \to j'} : [\MMM_i \to \NNN_j] \to [\MMM_{i'} \to \NNN_{j'}] &\quad f \mapsto \emb{j}{j'} \circ f \circ \proj{i'}{i},\\
\proj{i' \to j'}{i \to j} : [\MMM_{i'} \to \NNN_{j'}] \to [\MMM_i \to \NNN_j] &\quad f' \mapsto \proj{j'}{j} \circ f' \circ \emb{i}{i'}.
\end{align*}

\begin{proposition}
\label{funspprop}
If $\MMM_\III$ and $\NNN_\JJJ$ are both factor systems, so is their function space. Moreover,
\begin{enumerate}
\item If $\pmap$ satisfies (\ref{idmapeq}) or (\ref{stabeq}) on $\NNN_\JJJ$, so does $\pmap$ on $[\MMM_\III \to \NNN_\JJJ]$.
\item If $\MMM_\III$ is inverse and $\NNN_\JJJ$ direct, then $[\MMM_\III \to \NNN_\JJJ]$ is direct.
\item If $\MMM_\III$ is direct and $\NNN_\JJJ$ inverse, then $[\MMM_\III \to \NNN_\JJJ]$ is inverse.
\end{enumerate}
\end{proposition}

\begin{proof}
This has been proven in \cite{eberl2023}, except for the statements about (\ref{stabeq}) and about direct and inverse factor systems. Assume $f' \pmap f \comp g$ for $f' \in [\MMM_{i'} \to \NNN_{j'}]$ and $f, g \in [\MMM_{i} \to \NNN_{j}]$, $i \to j \leq i' \to j'$, so that we have to show $f' \pmap g$. Let $a_{i'} \pmap a_i$, and to confirm that $f'(a_{i'}) \pmap g(a_i)$, use the stability condition on $\NNN_\JJJ$ applied to $f'(a_{i'}) \pmap f(a_i) \comp g(a_i)$.

Next, let $\MMM_\III$ be inverse and $\NNN_\JJJ$ direct. Assume first that $f' \pmap f$ with $f' \in [\MMM_{i'} \to \NNN_{j'}]$, $f \in [\MMM_{i} \to \NNN_{j}]$ and $i \to j \leq i' \to j'$. We wish to show that $f' \comp \emb{i \to j}{i' \to j'}(f)$, which is the same as $f' \comp \emb{j}{j'} \circ f \circ \proj{i'}{i}$. Let $a_{i'} \pmap b_{i'}$ and define $b_i := \proj{i'}{i}(b_{i'})$. Then $\proj{i'}{i}(a_{i'}) \comp b_i$ since projections preserve $\comp$. It follows that $a_{i'} \pmap b_i$ since $\MMM_\III$ is inverse, and thus $f'(a_{i'}) \pmap f(b_i)$. Consequently, $f'(a_{i'}) \comp \emb{j}{j'}(f(b_i)) = \emb{j}{j'} \circ f \circ \proj{i'}{i}(b_{i'})$, since $\NNN_\JJJ$ is direct.

For the other direction assume $f' \comp \emb{j}{j'} \circ f \circ \proj{i'}{i}$ and let $a_{i'} \pmap a_i$. We shall prove that $f'(a_{i'}) \pmap f(a_i)$. Then $a_{i'} \pmap a_i$ implies $\proj{i'}{i}(a_{i'}) \comp a_i$ since $\MMM_\III$ is inverse. Now $f$ and the embeddings preserve $\comp$, so $\emb{j}{j'}(f(\proj{i'}{i}(a_{i'}))) \comp \emb{j}{j'}(f(a_i))$. Certainly $f'(a_{i'}) \comp \emb{j}{j'}(f(a_i))$. This shows $f'(a_{i'}) \pmap f(a_i)$, since $\NNN_\JJJ$ is direct, as claimed. The proof of the last statement is verified in a similar way.
\end{proof}

\subsection{Targets and Limits}
\label{tarsec}

The subsequent concepts require sets $\fil(\III)$ from Section \ref{forminf} with Property (\ref{filtereq}), whereby we write $\fil_\rho$ for $\fil(\III_\rho)$ and $\fil_\Gamma$ for $\fil(\III_\Gamma)$. A \emph{target} $(\MMM,\Pmap)$ for a system $\MMM_\III$ extends the system ``at the top'', i.e., the extension leads to a system $\MMM_{\bar\III}$, called \emph{compactification} of $\MMM_\III$, with index set $\bar\III := \III \cup \{top\}$, $top$ as greatest index, and $\MMM_{top} = \MMM$. Let $a \Pmap a_i$ denote $a \pmap a_i$ provided that $a \in \MMM$, and we also write $Emb_i$ for $\emb{i}{top}$ and $Proj_i$ for $\proj{top}{i}$. Relation $\pmap$ on a target $\MMM$, and consequently relation $\comp$ on $\MMM$ as well, is by definition the identity. The \emph{extension} of an element $a \in \MMM$ is $Ext(a) := \{a_i \in \bigcup_{i \in \III} \MMM_i \mid a \Pmap a_i\}$. A target $\MMM$ for a system $\MMM_\III$ satisfies by definition
\[
\III_a := \{i \in \III \mid \exists a_i \in \MMM_i \ a \Pmap a_i\} \in \fil(\III)
\]
for all objects $a \in \MMM$. Moreover, if the system is a prefactor system, a prefactor system with embeddings/projections or a factor system, then the compactification $\MMM_{\bar\III}$ must have this additional structure with its properties as well. Whereas the compactification of a system is automatically a system, the compactification of a \emph{prefactor system} requires for all $a \in \MMM$, $a_i \in \MMM_i$, $a_{i'} \in \MMM_{i'}$ and $i \leq i'$
\begin{equation}
\label{ppmapeq}
a \Pmap a_{i'}, \, a \Pmap a_i \ \imp \ a_{i'} \pmap a_i.
\end{equation}

If $\MMM$ is a target for a \emph{prefactor system with embeddings} $\MMM_{\III}$, then this implies the existence of $\comp$-embeddings $Emb_i : \MMM_i \to \MMM$, satisfying $Emb_{i'}(\emb{i}{i'}(a_i)) = Emb_{i}(a_i)$ and 
\begin{equation}
\label{embtareq}
a_{i'} \pmap a_i \ \imp\ Emb_{i'}(a_{i'}) \Pmap a_i \text{ for all } i \leq i',
\end{equation}
$a_i \in \MMM_i$, $a_{i'} \in \MMM_{i'}$ and $a \in \MMM$. If $\MMM$ is a target for a \emph{prefactor system with projections} $\MMM_{\III}$, then there are moreover $\comp$-projections $Proj_i : \MMM \to \MMM_i$ such that $\proj{i'}{i}(Proj_{i'}(a)) \comp Proj_i(a)$ and 
\begin{equation}
\label{projtareq}
a \Pmap a_i \ \imp\ Proj_{i'}(a) \pmap a_i \text{ for all } i \leq i'.
\end{equation}

If $\MMM$ is a target for a \emph{factor system}, then $Emb$ and $Proj$ with these properties exist and both form a $\comp$-embedding-projection pair. We write $(\MMM,\Pmap)$, $(\MMM,\Pmap,Emb)$, $(\MMM,\Pmap,Proj)$ and $(\MMM,\Pmap,Emb,Proj)$, respectively, for these targets. If $\MMM$ is a target for a prefactor system with projections $\MMM_\III$, then the projections $Proj_i$ can be defined by 
\begin{equation}
\label{Projdef}
Proj_i(a) := \proj{i'}{i}(a_{i'}) \text{ for some } i' \geq i \text{ with } a \Pmap a_{i'}
\end{equation}
for $a \in \MMM$. It follows from the properties of a prefactor system that $Proj_i(a)$ is unique modulo $\comp$. Since $Emb_{i'}(a_{i'}) \Pmap a_{i'}$ for $a_{i'} \in \MMM_{i'}$ by (\ref{embtareq}), we have for $i \leq i'$
\begin{equation}
\label{Projprop}
Proj_i(Emb_{i'}(a_{i'})) \comp \proj{i'}{i}(a_{i'}).
\end{equation}

A target $(\MMM,\Pmap)$ for a system $\MMM_\III$ is a \emph{limit} of $\MMM_\III$ iff for every further target $(\NNN,\Qmap)$ for $\MMM_\III$ there is a unique map $\Phi : \NNN \to \MMM$ such that $a \Qmap a_i$ implies $\Phi(a) \Pmap a_i$. If the underlying system is a factor system or a prefactor system, then we call the limit \emph{factor limit} and \emph{prefactor limit}, resp. It turns out, however, that a factor limit is the same as the prefactor limit, and that this limit $\limf(\MMM_\III)$ is unique modulo isomorphism. Therefore we simply speak of a \emph{limit}, or a \emph{limit set}, if we want to distinguish it from a \emph{limit element} in this limit set.

\begin{example}
\label{dirinvex2}
Recall Example \ref{dirinvex}. If a direct system is also a direct factor system and $a_{i'} \pmap a_i \iffdef \emb{i}{i''}(a_{i}) = \emb{i'}{i''}(a_{i'})$ for some $i'' \geq i, i'$, then the factor limit is the direct limit, i.e., $\limf(\MMM_\III) = \dirlim(\MMM_\III)$. If an inverse system is additionally an inverse factor system and $a_{i'} \pmap a_i \iffdef \proj{i'}{i}(a_{i'}) = a_i$, then the factor limit $\limf(\MMM_\III)$ is the inverse limit $\invlim(\MMM_\III)$.

In the first case of Example \ref{natex} the limit is $\nat$ (which is the direct limit as well as the factor limit). In the second case the limit (inverse limit and factor limit) is $\nat_\infty := \nat \cup \{\infty\}$ with $\infty = (0,1,2,\dots)$. In the third case the factor limit is a singleton set, which is neither the direct nor the inverse limit.
\end{example}

\subsection{Consistent Sets and Dynamic Elements}
\label{limsec}

It is possible to define concrete targets and limits in the form of sets of states: A set $\alpha \subseteq \bigcup_{i \in \III} \MMM_i$ in a system $\MMM_\III$ is called a \emph{consistent set} iff $a_{i'} \pmap a_i$ holds for all $a_{i'}, a_i \in \alpha$ with $i' \geq i$, and $\III_\alpha := \{i \in \III \mid \alpha \cap \MMM_i \not= \emptyset\} \in \fil(\III)$. If we already have a target $\MMM$ for a prefactor system, then the set $Ext(a)$ is such a consistent set for all $a \in \MMM$. The set of all consistent sets in a (pre)factor system $\MMM_\III$ is itself a target for $\MMM_\III$ with $\alpha \Pmap a_i \iffdef a_i \in \alpha$. In other words, if $\MMM$ denotes the set of all consistent sets, then the target is $(\MMM, \ni)$. 

A \emph{dynamic element} is a maximal (w.r.t.~subset inclusion) consistent set, and for each consistent set $\alpha$ in a prefactor system there is exactly one dynamic element $\alpha^m$ such that $\alpha \subseteq \alpha^m$. Let $\elem(\MMM_\III)$ denote the set of all of these dynamic elements, then 
\begin{equation}
\label{elemtergeteq}
(\elem(\MMM_\III),\ni)
\end{equation}
is a prefactor limit of $\MMM_\III$. For the next two lemmata, recall that it is assumed that $\fil(\III)$ is a filter as defined in (\ref{filtereq}).

\begin{lemma}
\label{elemlem}
Let $\alpha$ be a consistent set in a prefactor system $(\MMM_\III,\pmap)$ and $b_i \in \MMM_i$, then the following are equivalent:
\begin{enumerate}
\item \label{1elemlem} $b_i \in \alpha^m$.
\item \label{2elemlem} $\alpha \cup \{b_i\}$ is a consistent set.
\item \label{3elemlem} There are cofinal many $i' \in \III$ with $a_{i'} \pmap b_i$ for some $a_{i'} \in \alpha$.
\item \label{4elemlem} $a_{i'} \pmap b_i$ for all $a_{i'} \in \alpha$ with $i' \geq i$.
\item \label{5elemlem} There are $\fil$-many $i' \in \III$ with $a_{i'} \pmap b_i$ for some $a_{i'} \in \alpha$.
\end{enumerate}
\end{lemma}

\begin{proof}
This has been shown in \cite{eberl2023}, except the last clause. For its equivalence to the other statements notice that it follows from Clause \ref{4elemlem}.~since $\III_\alpha \, \cap \!\up i \in \fil(\III)$, and it implies Clause \ref{3elemlem}.~since each set in $\fil(\III)$ is cofinal.
\end{proof}

\begin{lemma}
\label{elemlem2}
Let $\alpha$ and $\beta$ be consistent sets in a prefactor system $(\MMM_\III,\pmap)$, then the following are equivalent:
\begin{enumerate}
\item \label{1elemlem2} $\alpha^m = \beta^m$.
\item \label{2elemlem2} $a_i \comp b_i$ for all $i \in \III$ with $a_i \in \alpha$ and $b_i \in \beta$.
\item \label{3elemlem2} $a_i \comp b_j$ for all $i,j \in \III$ with $a_i \in \alpha$ and $b_j \in \beta$.
\end{enumerate}
\end{lemma}

\begin{proof}
The equivalence of \ref{1elemlem2}.~and \ref{3elemlem2}.~has been shown in \cite{eberl2023}, so it suffices to prove that \ref{2elemlem2}.~implies $b_i \in \alpha^m$ for each $b_i \in \beta$. We wish to find cofinal many $a_{i'} \in \alpha$ with $a_{i'} \pmap b_i$ and apply Lemma \ref{elemlem}. There are $\fil$-many $i' \geq i$ with $i' \in \III_\alpha \cap \III_\beta$. For all $a_{i'} \in \alpha$, $b_{i'} \in \beta$ we have $a_{i'} \comp b_{i'}$ by assumption, hence $a_{i'} \pmap b_i$, since $b_{i'} \pmap b_i$.
\end{proof}

If one, and hence all, of the conditions in Lemma \ref{elemlem2} are true, then we write $\alpha \sim \beta$. Call an element $a \in \MMM$ of a target $\MMM$ for a prefactor system $\MMM_\III$ a \emph{limit (element)} of a consistent set $\alpha$ iff $\alpha \sim Ext(a)$. Obviously, $a$ is a limit of its extension $Ext(a)$. If there is only one limit element, we denote it as $lim(\alpha)$. 

Given a factor system $\MMM_\III$. Then $\elem(\MMM_\III)$ is, up to isomorphism, the limit $\limf(\MMM_\III)$, whereby projections have been defined by (\ref{Projdef}) and embeddings are
\begin{align}
\label{Embdef}
Emb_i(a_i) &:= \{\emb{i}{i'}(a_{i}) \in \bigcup_{j \in \III} \MMM_j \mid i' \geq i\}^m.
\end{align} 

Let $\MMM$ be a target for a system $\MMM_\III$, then $\MMM$ is \emph{maximal (over $\MMM_\III$)} iff $Ext(a) \in \elem(\MMM_\III)$ for all $a \in \MMM$. $\MMM$ is \emph{extensional (over $\MMM_\III$)} iff $Ext(a) \sim Ext(b)$ implies $a = b$ for all $a, b \in \MMM$. $\MMM$ is \emph{complete (over $\MMM_\III$)} iff for all consistent sets $\alpha$ there is a limit element $a$, i.e., some $a \in \MMM$ with $Ext(a) \sim \alpha$. One can characterize a limit $\limf(\MMM_\III)$ also as a target that is maximal, extensional and complete over $\MMM_\III$. 

\begin{proposition}
\label{invdieprop}
Assume a prefactor system $\MMM_\III$ has been compactified with a target $\MMM$, yielding the extended prefactor system $\MMM_{\bar\III}$.
\begin{enumerate}
\item If $\MMM_\III$ is stable and $\MMM$ maximal over $\MMM_\III$, then $\MMM_{\bar\III}$ is stable, too.
\item If $\MMM_\III$ is direct and $\MMM$ extensional over $\MMM_\III$, then $\MMM_{\bar\III}$ is direct and $\III_a$ contains an up-set for all $a \in \MMM$.
\item If $\MMM_\III$ is inverse and $\MMM$ maximal over $\MMM_\III$, then $\MMM_{\bar\III}$ is inverse and $\III_a = \III$ for all $a \in \MMM$.
\end{enumerate}
\end{proposition}

\begin{proof}
1.~In order to show stability of $\MMM_{\bar\III}$ it suffices to prove that $a \Pmap a_i \comp b_i$ implies $a \Pmap b_i$ for all $a \in \MMM$. Indeed, there are cofinal many indices $i' \geq i$ with $a \Pmap a_{i'}$, and consequently $a_{i'} \pmap a_i$. Applying (\ref{stabeq}) yields $a_{i'} \pmap b_i$ for cofinal many indices $i' \in \III$. By Lemma \ref{elemlem} $a \Pmap b_i$ follows, since $Ext(a)$ is a maximal consistent set.

2.~We claim that $a \Pmap a_i \iff a = Emb_i(a_i)$ for $a \in \MMM$. So assume $a \Pmap a_i$ and we first establish that $a \Pmap a_j$ implies $Emb_i(a_i) \Pmap a_j$ for all $j \in \III$. There are cofinal many $i' \geq i, j$ with $a \Pmap a_{i'}$, $a_{i'} \pmap a_i$ and $a_{i'} \pmap a_j$, hence $a_{i'} \comp \emb{i}{i'}(a_i)$, because $\MMM_\III$ is direct. It follows that $\emb{i}{i'}(a_i) \pmap a_j$ for cofinal many indices $i' \in \III$ and thus $Emb_i(a_i) \Pmap a_j$. This proves $Ext(a) \sim Ext(Emb_i(a_i))$, so $a$ and $Emb_i(a_i)$ are equal since $\MMM$ is extensional over $\MMM_\III$. The inverse implication $a = Emb_i(a_i) \imp a \Pmap a_i$ holds trivially. Moreover, $a = Emb_i(a_i)$ for $a \in \MMM$ implies that $\III_a$ contains the up-set $\! \up i$.

3.~We shall prove that $a \Pmap a_i \iff Proj_i(a) \comp a_i$. For the forward implication assume $Proj_i(a) = \proj{i'}{i}(a_{i'})$ for some $a_{i'}$ with $a \Pmap a_{i'}$. We establish that $\proj{i'}{i}(a_{i'}) \comp a_i$ for $a \Pmap a_i$ by using the fact that $\MMM_\III$ is inverse and that $a \Pmap a_{i'}$ and $a \Pmap a_i$ implies $a_{i'} \pmap a_i$.

For the backward implication assume $Proj_i(a) = \proj{i'}{i}(a_{i'}) \comp a_i$. To show $a \Pmap a_i$ it suffices to find cofinal many indices $i'' \geq i'$ with $a \Pmap a_{i''}$ and $a_{i''} \pmap a_i$. To this aim, we use maximality of $\MMM$. Indeed, there are cofinal many indices $i'' \geq i'$ with $a \Pmap a_{i''}$, and to confirm that $a_{i''} \pmap a_i$ it suffices to prove $\proj{i''}{i}(a_{i''}) \comp a_i$. Now $a_{i''} \pmap a_{i'}$ is a consequence of $a \Pmap a_{i''}$ and $a \Pmap a_{i'}$, hence $\proj{i''}{i'}(a_{i''}) \comp a_{i'}$. Because $\comp$-projections preserve $\comp$, we deduce $\proj{i''}{i}(a_{i''}) \comp \proj{i'}{i}(\proj{i''}{i'}(a_{i''})) \comp \proj{i'}{i}(a_{i'}) \comp a_{i}$, as claimed. For all $i \in \III$ there is some $a_i$ with $a \Pmap a_i \comp Proj_i(a)$, hence $\III_\alpha = \III$.
\end{proof}

\begin{corollary}
\label{invdiecor}
Assume a prefactor system $\MMM_\III$ has been compactified with a limit $\MMM$, yielding the extended prefactor system $\MMM_{\bar\III}$. If $\MMM_\III$ is stable (direct, inverse), then $\MMM_{\bar\III}$ is also stable (direct, inverse resp.).
\end{corollary}

\subsection{Targets and Limits on the Function Space}
\label{moretlsec}

Consider two factor systems $\MMM_\III$ and $\NNN_\JJJ$ with limits $\MMM$ and $\NNN$ resp. Then 
\[
[\MMM \to_{\fil} \NNN] := \{f : \MMM \to \NNN \mid \III_f \in \fil(\III \times \JJJ)\}
\]
is a target for $[\MMM_\III \to \NNN_\JJJ]$, but not necessarily a limit. We now state a condition which guarantees that $[\MMM \to_{\fil} \NNN]$ is indeed a limit. So let $\zeta$ be a consistent set on the factor system $[\MMM_\III \to \NNN_\JJJ]$, and $\alpha$ a consistent set on the factor system $\MMM_\III$. Define
\begin{equation}
\label{appcond}
\tag{Appl}
App(\zeta, \alpha) := \zeta(\alpha) := \{f_{i \to j}(a_i) \mid f_{i \to j} \in \zeta \text{ and } a_i \in \alpha\}.
\end{equation}

It is easy to see that this set satisfies $b_{j'} \pmap b_j$ for all $j \leq j'$ and $b_j, b_{j'} \in \zeta(\alpha)$. However, $\zeta(\alpha)$ is not necessarily a consistent set. For that, $\III_f[\III_a] \in \fil(\JJJ)$ must be the case, whereby $\HHH[\III'] := \{j \in \JJJ \mid \exists i \in \III' \ i \to j \in \HHH\}$ for $\HHH \subseteq \III \times \JJJ$ and $\III' \subseteq \III$. This is ensured when the following condition for $\fil$ is met for all sets $\III$ and $\JJJ$:
\begin{equation}
\label{Deq}
\tag{D}
\HHH \in \fil(\III \times \JJJ) \text{ and } \III' \in \fil(\III) \text{ implies } \HHH[\III'] \in \fil(\JJJ).
\end{equation}

We will require Condition (\ref{Deq}) for the rest of this paper.

\begin{lemma}
\label{appextlem}
The application defined by (\ref{appcond}) is extensional, i.e., for all $\zeta, \zeta' \in \elem([\MMM_\III \to \NNN_\JJJ])$ we have:
\[
\zeta(\alpha) = \zeta'(\alpha) \text{ for all } \alpha \in \elem(\MMM_\III) \ \text{ implies } \ \zeta = \zeta'.
\]
\end{lemma}

\begin{proof}
Assume $\zeta(\alpha) = \zeta'(\alpha)$ for all $\alpha \in \elem(\MMM_\III)$ and let $f_{i \to j} \in \zeta$. We will show that $f_{i \to j} \in \zeta'$, which proves the lemma. Choose indices $i' \to j' \in \III_\zeta \cap \, \III_{\zeta'} \, \cap \! \up(i \to j) \in \fil(\III \times \JJJ)$ and let $f'_{i' \to j'} \in \zeta'$. We wish to show $f'_{i' \to j'} \pmap f_{i \to j}$. Since this is the case for $\fil$-many indices $i' \to j'$, we have shown $f_{i \to j} \in \zeta'$ by Lemma \ref{elemlem}, which suffices.

Now assume $a_{i'} \pmap a_i$ for $a_{i'} \in \MMM_{i'}$ and $a_i \in \MMM_i$, so that we have to prove $f'_{i' \to j'}(a_{i'}) \pmap f_{i \to j}(a_i)$. For $\alpha := Emb_{i'}(a_{i'}) \in \elem(\MMM_\III)$ we know that $a_i \in \alpha$, by maximality and (\ref{embtareq}). Further, $f'_{i' \to j'}(a_{i'}) \in \zeta'(\alpha) = \zeta(\alpha)$ and $f_{i \to j}(a_i) \in \zeta(\alpha)$ implies $f'_{i' \to j'}(a_{i'}) \pmap f_{i \to j}(a_i)$, because $\zeta(\alpha)$ is a consistent set.
\end{proof}

\begin{proposition}
\label{eqcondprop}
Let $\MMM$ and $\NNN$ be a targets for the factor systems $\MMM_\III$ and $\NNN_\JJJ$ resp. Assume Condition (\ref{Deq}) holds. If $\NNN$ is maximal (extensional, complete) over $\NNN_\JJJ$, so is $[\MMM \to_\fil \NNN]$ over $[\MMM_\III \to \NNN_\JJJ]$.
\end{proposition}

\begin{proof}
First we shall prove the claim about maximality, so let $g_{i \to j} \in \MMM_{i \to j}$ and we apply Lemma \ref{elemlem}. Assume there is a set $\HHH \in \fil(\III \times \JJJ)$ with $f \Pmap f_{i' \to j'} \pmap g_{i \to j}$ for $i' \to j' \in \HHH$, then it suffices to show $f \Pmap g_{i \to j}$. By the definition of $\Pmap$ it must be checked that 
\[
f(a) \Pmap g_{i \to j}(a_i) \ \text{ for } \ a \Pmap a_i.
\]

There are $\fil$-many indices $i' \geq i$ with $a \Pmap a_{i'}$, since $\III_a \, \cap \! \up i \in \fil(\III)$. They satisfy $a_{i'} \pmap a_i$, because $a \Pmap a_{i'}$ and $a \Pmap a_i$. Applying Condition (\ref{Deq}) give us $\fil$-many indices $j' \in \HHH[\III_a \, \cap \! \up i]$, which all satisfy $f(a) \Pmap f_{i' \to j'}(a_{i'}) \pmap g_{i \to j}(a_i)$. So maximality of $\NNN$ yields $f(a) \Pmap g_{i \to j}(a_i)$, showing the maximality of $[\MMM \to_\fil \NNN]$.

For extensionality we prove that $f = g$, provided $Ext(f) \sim Ext(g)$. So let $a \in \MMM$ and we claim that $f(a) = g(a)$. For all $i \to j \in \III_f \cap \III_g$ and $f_{i \to j} \in Ext(f)$, $g_{i \to j} \in Ext(g)$ we have $f_{i \to j}\comp g_{i \to j}$ by definition. Condition (\ref{Deq}) shows that 
\[
\JJJ' := (\III_f \cap \III_g)[\III_a] \in \fil(\JJJ),
\]
and for all these indices $i \in \III_a$ and $j \in \JJJ'$ we have $a \Pmap a_i$ for some $a_i \in \MMM_i$,  $f(a) \Pmap f_{i \to j}(a_i)$, $g(a) \Pmap g_{i \to j}(a_i)$ and $f_{i \to j}(a_i) \comp g_{i \to j}(a_i)$. This yields $Ext(f(a)) \sim Ext(g(a))$, since the elements of $Ext(f(a))$ and $Ext(g(a))$ are consistent (i.e., in $\comp$-relation to each other) for cofinal many indices $j \in \JJJ$. We now appeal to extensionality of $\NNN$ to get $f(a) = f(b)$.

To show completeness of $[\MMM \to_\fil \NNN]$ over $[\MMM_\III \to \NNN_\JJJ]$ let $\zeta$ be a consistent set in $[\MMM_\III \to \NNN_\JJJ]$. We define a limit element $f : \MMM \to \NNN$ of $\zeta$ by
\[
f(a) := \text{ limit element of } \beta := \{g_{i \to j}(a_i) \mid g_{i \to j} \in \zeta \cap [\MMM_i \to \NNN_j] \text{ and } a \Pmap a_i\}
\]
for all $a \in \MMM$. This function is well-defined: If $a \Pmap a_i$, $a \Pmap b_i$ and $f_{i \to j}, g_{i \to j} \in \zeta \cap [\MMM_i \to \NNN_j]$, then $a_i \comp b_i$ follows and consequently $f_{i \to j}(a_i) \comp g_{i \to j}(b_i)$; note that $a_i \comp b_i \iff a_i \pmap b_i$ and the same for $f_{i \to j}(a_i) \comp g_{i \to j}(b_i)$. This implies that any two elements in $\beta$ are consistent. Condition (\ref{Deq}) guarantees that $\III_f = \III_\zeta[\III_a] \in \fil(\JJJ)$. The set $\beta$ is thus a consistent set and completeness of $\NNN$ ensures that a limit element of $\beta$ exists.

We claim that $f$ is a limit of $\zeta$. Let $f \Pmap f_{i \to j}$ and $g_{i \to j} \in \zeta$ with $f_{i \to j}, g_{i \to j} \in \MMM_{i \to j}$ and we have to show $f_{i \to j} \pmap g_{i \to j}$. Given $b_i \pmap a_i$, both in $\MMM_i$, so it suffices to show that $f_{i \to j}(b_i) \pmap g_{i \to j}(a_i)$. Take $a := Emb_i(a_i) = Emb_i(b_i)$, then $a \Pmap a_i$ and $a \Pmap b_i$, hence $f(a) \Pmap f_{i \to j}(b_i)$ and $g_{i \to j}(a_i) \in \beta$. Since $f(a)$ is a limit element of $\beta$, it follows that $\beta \sim Ext(f(a))$, and thus $f_{i \to j}(b_i) \comp g_{i \to j}(a_i)$ by Lemma \ref{elemlem2}, so we are done.
\end{proof}

Recall that limits are targets which are maximal, extensional and complete. Recall also that $\elem([\MMM_\III \to \NNN_\JJJ])$ is a limit of $[\MMM_\III \to \NNN_\JJJ]$ and that limits are unique modulo isomorphism.

\begin{corollary}
\label{limitfuncor}
Given two factor systems $\MMM_\III$ and $\NNN_\JJJ$ with limits $\MMM$ and $\NNN$ resp. Assume Condition (\ref{Deq}) holds, then $[\MMM \to_{\fil} \NNN]$ is a limit of $[\MMM_\III \to \NNN_\JJJ]$. Equivalently stated, there is an isomorphism $[\elem(\MMM_\III) \to_{\fil} \elem(\NNN_\JJJ)] \simeq \elem([\MMM_\III \to \NNN_\JJJ])$.
\end{corollary}

\section{Towards Models and Interpretations}
\label{sysemlam}

This section contains the first steps towards an interpretation of a typed $\lambda$-term in a factor system and a limit of it. Let types $\rho \in \type$ be given by $\rho ::= \iota \,\,|\,\, \rho \to \rho$ for base types $\iota$. We assume that type $\prop$ is one of these base types. For each type $\rho$ let a non-empty, directed set $\III_\rho$ of indices be given and let $\fil_\Gamma := \fil(\III_\Gamma)$ be a filter which satisfies (\ref{filtereq}). Maps on indices, such as $Emb_i$, extend in an obvious way to state contexts $C = (i_0, \dots, i_{n-1})$ as products, e.g.~$Emb_C = Emb_{i_0} \times \dots \times Emb_{i_{n-1}}$.

Given also a basic version of a simply typed $\lambda$-calculus (without constants) and for convenience we use a fixed list of variables $\yx{0},\yx{1},\yx{2},\dots$ within the Church-style $\lambda$-terms $\y{r} ::= \yx{k} \,\,|\,\, \y{r}\y{r} \,\,|\,\, \lambda \yxt{k}{\rho} \y{r}$. We apply the usual conventions for $\lambda$-terms \cite{nederpelt2014type}.

\subsection{Judgements for Types and States}
\label{typdecsec}

Typing judgements $\tdt{\y{r}}{\Gamma}{\rho}$ with $\Gamma = (\rho_0, \dots, \rho_{n-1})$ are defined recursively as follows:
\begin{align*}
&\qquad \inferrule*[Left = (Var)]{k < n}{\tdt{\yx{k}}{\Gamma}{\rho_k}} \qquad \qquad
\inferrule*[Left = (App)]{\tdt{\y{r}}{\Gamma}{\rho \to \sigma} \\ \tdt{\y{s}}{\Gamma}{\rho}}{\tdt{\y{rs}}{\Gamma}{\sigma}} \qquad \qquad \inferrule*[Left = (Abs)]{\tdt{\y{r}}{\Gamma.\rho}{\sigma}}{\tdt{\lambda \yxt{n}{\rho} \y{r}}{\Gamma}{\rho \to \sigma}} 
\end{align*}

We can write more explicitly $(\yx{0} : \rho_0, \dots, \yx{n-1} : \rho_{n-1}, \yx{n} : \rho_n)$ for the context $\Gamma.\rho_n$, so $\lambda$-abstractions binds the last variable of the context.\footnote{This style is basically the de Bruijn level notation of $\lambda$-terms, see e.g.~\cite{Plotkin1999}.} When we speak of typed terms this refers to this typing judgement. It is well known that each term $\y{r}$ in the simply typed $\lambda$-calculus has a unique type for a given context $\Gamma$.

\begin{definition}
\label{posnegdef}
Let $\iota$ denote a base type. The \emph{positive} and \emph{negative} types are defined as 
\[
\type^+ \ni \rho^+ ::= \iota \,\,|\,\, \rho^- \to \rho^+ \quad \text{ and } \quad
\type^- \ni \rho^- ::= \prop \,\,|\,\, \rho^+ \to \rho^-.
\]

For an index $i \in \III_\rho$ with $\rho \in \type^+$ we write $i \in Idx^+$, and similar for $i \in Idx^-$. Moreover, $\type^c \ni \rho^c ::= \iota \,\,|\,\, \rho^+ \to \rho^c \,\,|\,\, \rho^- \to \rho^c$.
\end{definition}

For example, $nat$ is a positive type, $nat \to nat \to \prop$ is a negative type, $\prop \to \prop$ is both and $nat \to nat$ is neither. Obviously $\type^+ \cup \type^- \subseteq \type^c$ and $\rho \in \type^{+} \cap \type^{-} \iff \rho ::= \prop \ |\  \rho \to \rho$. In order to formulate the rules for states, we use a fragment of the simply typed $\lambda$-calculus, which we refer to as the \emph{core fragment}. This fragment is constrained in such a way that contexts $\Gamma = (\rho_0, \dots, \rho_{n-1})$ contain only positive or negative types, that is, 
\begin{equation}
\label{restconteq}
\rho_k \in \type^+ \cup \type^- \text{ for all } \rho_k \text{ in } \Gamma.
\end{equation}

A more general definition would permit the use of further base types in place of the negative type $\prop$. This requires that the base type in question has a finite set of objects and an interpretation analogous to that of type $\prop$ in Definition \ref{typedef}. The primary motivation for this definition is the existence of suitable rules for state judgements of variables and the fact that positive types can be interpreted as direct factor systems and negative types as inverse factor systems.

\begin{lemma}
The typed term $\tdt{\y{r}}{\Gamma}{\rho}$ is in the core fragment iff $\Gamma \to \rho \in \type^c$.
\end{lemma}

This lemma follows easily by induction on $\tdt{\y{r}}{\Gamma}{\rho}$. Therein $\Gamma \to \sigma$ is defined by $\nil \to \sigma := \sigma$ and $\Gamma.\rho \to \sigma := \Gamma \to (\rho \to \sigma)$; the index $C \to i$ is defined in the same way as $\Gamma \to \sigma$. We now introduce judgements for states of the form $\tdts{\y{r}}{C}{i}$. They require a typing judgement $\tdt{\y{r}}{\Gamma}{\rho}$ such that $C \to i \in \III_{\Gamma \to \rho}$. The rules are based on the typing rules\footnote{One may also combine both rules into one. The rules are however more readable if one keeps them separated.}, whereby $C = (i_0, \dots,i_{n-1})$.
\begin{align*}
&&&\inferrule*[Left = (Var$+$)]{j \geq i_k \in Idx^+}{\tdts{\yx{k}}{C}{j}} &&\inferrule*[Left = (Var$-$)]{j \leq i_k \in Idx^-}{\tdts{\yx{k}}{C}{j}} \\
&&&\inferrule*[Left = (App)]{\tdts{\y{r}}{C}{i \to j} \\ \tdts{\y{s}}{C}{i}}{\tdts{\y{rs}}{C}{j}} &&\inferrule*[Left = (Abs)]{\tdts{\y{r}}{C.i}{j}}{\tdts{\lambda \yxt{n}{\rho} \y{r}}{C}{i \to j}} 
\end{align*}

The condition $j \geq i_k$ in Rule (\textsc{Var$+$}) is a consequence of the fact that objects $a_{i_k}$ of positive type exist at all later stages $j$ as $\emb{i_k}{j}(a_{i_k})$. The condition $j \leq i_k$ in Rule (\textsc{Var$-$}) is related to the fact that we find unique restrictions of a relation, but not unique extensions. Note that for an index $i_k \in Idx^{+} \cap Idx^{-}$ both of the rules for variables apply. In this case, however, there is only one index, so both rules are the same.

A term usually has several state judgement, not only one. Indeed, we expect that there are $\fil$-many of them and it seems necessary to prove this stronger property in order to guarantee that at least one judgement exists. However, we have not shown that all terms of the core fragment have a state judgement. It is possible to define sets $\fil_\Gamma$, similar as in \cite{eberl2022RML}, but with a reorder of the context $\Gamma$ that takes the positive types first, and afterward the negative ones. In a subsequent paper we will present such a definition of sets $\fil_\Gamma$ satisfying Condition (\ref{Deq}) and a proof that each term in the above mentioned fragment has indeed $\fil$-many state judgements.

\subsection{Interpretation of Types}
\label{typeintsec}

An interpretation of a type $\rho$ has a static and a dynamic part. Although the limit is uniquely defined from a structural perspective, the extension, i.e., the set of elements in the limit, depends on the stage of the meta-level investigation, assuming a consequent finitistic view. In this case, the limit is not the absolute end of the extensible system. As the investigation progresses, both the system and the limit increase. The main part of the interpretation is the factor system --- the limit is necessary to prove that the definition in the factor system is correct.

Let $(\val{\rho}_i)_{i \in \III_\rho}$ be the factor system that interprets type $\rho$, which consists of finite sets $\val{\rho}_i$. The limit, up to isomorphism, is then $\val{\rho} := \elem((\val{\rho}_i)_{i \in \III_\rho})$, see Section \ref{limsec}. These limit sets give rise to an extensional type structure if we use (\ref{appcond}) as application. Recall Lemma \ref{appextlem} and our assumption of Condition (\ref{Deq}).

\begin{definition}
\label{typedef}
An interpretation $\val{\ }$ of types assigns to a type $\rho$ a factor system $(\val{\rho}_i)_{i \in \III_\rho}$ and a limit $\val{\rho}$ of it. So the pair $((\val{\rho}_i)_{i \in \III_\rho}, \val{\rho})$ interprets the type $\rho$. This interpretation shall satisfy the following properties:
\begin{enumerate}
\item For a base type $\iota$, $(\val{\iota}_i)_{i \in \III_\iota}$ is a direct factor system.
\item Type $\prop$ is interpreted by the factor system $\bool_{\{\prop\}}$ with $\bool_{\prop} = \bool = \{true, false\}$ and limit $\bool$. 
\item The interpretation of $\rho \to \sigma$ is the factor system $[(\val{\rho}_i)_{i \in \III_\rho} \to (\val{\sigma}_j)_{j \in \III_\sigma}]$ together with the limit $[\val{\rho} \to_\fil \val{\sigma}]$.
\end{enumerate}
\end{definition}

Since Condition (\ref{Deq}) holds, it follows from Corollary \ref{limitfuncor} that the function space $[\val{\rho} \to_\fil \val{\sigma}]$ is indeed the limit of the underlying factor system $[(\val{\rho}_i)_{i \in \III_\rho} \to (\val{\sigma}_j)_{j \in \III_\sigma}]$. Definition \ref{typedef} extends to contexts $\Gamma$ in the usual way by taking products.

\begin{example}
The standard interpretation of the natural numbers $\val{nat} = \nat$ is the limit of the direct factor system $(\val{nat}_i)_{i \in \nat^+}$ with $\val{nat}_i = \nat_i$.\footnote{Here we assume that the interpretation of type $nat$ is the set $\nat$ of all natural numbers, which is an actual infinite set if one accepts actual infinities at meta-level. If not, $\nat$ is an arbitrary large finite set, which corresponds to a set in the implicitly given factor system at meta-level, i.e., $\nat$ is $\nat_j$ for some sufficiently large index $j$ depending on the stages of the meta-level investigation.} Of course, there are other non-standard models of $nat$ with non-standard natural numbers as well. These non-standard numbers do not appear at the limit step, but at some stage within the system.
\end{example}

The next lemma is required for the reflection principle stated in Theorem \ref{mainthm}. Its proof proceeds by induction on Definition \ref{posnegdef} and uses Proposition \ref{funspprop} and Corollary \ref{invdiecor}.

\begin{lemma}
\label{embpmaplem}
Given an interpretation of types $\val{\ }$ and let $i \leq i'$ with $i, i' \in \III_\rho$ for a type $\rho$. 
\begin{enumerate}
\item If $\rho \in \type^{+}$, then the factor system $(\val{\rho}_i)_{i \in \III_\rho}$ is direct and $a \Pmap a_i$ implies $a \Pmap \emb{i}{i'}(a_i)$.
\item If $\rho \in \type^{-}$, then the factor system is inverse and $a \Pmap a_{i'}$ implies $a \Pmap \proj{i'}{i}(a_{i'})$.
\end{enumerate}
\end{lemma}

\subsection{Interpretation of Terms and a First Version of the Reflection Principle}
\label{inttermprelimsec}

The interpretation of terms in the type structure of the limit sets is the common interpretation. What is new is the interpretation in the factor system. This part of the interpretation requires a state judgement, so it is defined only for the core fragment, defined in Section \ref{typdecsec}.

\begin{definition}
A $\Gamma$-environment is a list of elements $\vecb{A} \in \val{\Gamma} := \val{\rho_0} \times \dots \times \val{\rho_{n-1}}$ for $\Gamma = (\rho_0, \dots, \rho_{n-1})$. The value $\val{\y{r}}_{\vecb{A}} \in \val{\rho}$ of a typed term $\tdt{\y{r}}{\Gamma}{\rho}$ is defined recursively on the derivation of the judgement $\tdt{\y{r}}{\Gamma}{\rho}$ relative to a $\Gamma$-environment $\vecb{A} \in \val{\Gamma}$:
\begin{align*}
\val{\yx{k}}_{\vecb{A}} &:= A_k &&\text{for } \tdt{\yx{k}}{\Gamma}{\rho_k}, \\
\val{\y{rs}}_{\vecb{A}} &:= \val{\y{r}}_{\vecb{A}}(\val{\y{s}}_{\vecb{A}}) && \text{for } \tdt{\y{r}}{\Gamma}{\rho \to \sigma} \text{ and } \tdt{\y{s}}{\Gamma}{\rho},\\
\val{\lambda \yxt{n}{\rho} \y{r}}_{\vecb{A}}(B) &:= \val{\y{r}}_{\vecb{A}.B} && \text{for } \tdt{\y{r}}{\Gamma.\rho}{\sigma} \text{ and } B \in \val{\rho}.
\end{align*}
\end{definition}

\begin{definition}
\label{facttermdef}
A $C$-environment is a list of elements $\vecb{a} \in \val{\Gamma}_C := \val{\rho_0}_{i_0} \times \dots \times \val{\rho_{n-1}}_{i_{n-1}}$ for $C = (i_0, \dots, i_{n-1})$. The value $\val{\y{r}}_{\vecb{a} : C}^i \in \val{\rho}_i$ of a term with $\tdts{\y{r}}{C}{i}$ relative to a $C$-environment $\vecb{a} \in \val{\Gamma}_C$ is defined recursively on the derivation of the state judgement $\tdts{\y{r}}{C}{i}$:
\begin{align*}
\val{\yx{k}}_{\vecb{a} : C}^j &:= \emb{i_k}{j}(a_k) &&\text{for } \tdts{\yx{k}}{C}{i_k} \text{ and } \rho_k \in \type^+, \\
\val{\yx{k}}_{\vecb{a} : C}^j &:= \proj{i_k}{j}(a_k) &&\text{for } \tdts{\yx{k}}{C}{i_k} \text{ and } \rho_k \in \type^-, \\
\val{\y{rs}}_{\vecb{a} : C}^j &:= \val{\y{r}}_{\vecb{a} : C}^{i \to j}(\val{\y{s}}_{\vecb{a} : C}^i) && \text{for } \tdts{\y{r}}{C}{i \to j} \text{ and } \tdts{\y{s}}{C}{i},\\
\val{\lambda \yxt{n}{\rho} \y{r}}_{\vecb{a} : C}^{i \to j}(b) &:= \val{\y{r}}_{\vecb{a}.b : C.i}^j && \text{for } \tdts{\y{r}}{C.i}{j} \text{ and } b \in \val{\rho}_i.
\end{align*}
\end{definition}

By definition, the value $\val{\y{r}}_{\vecb{a} : C}^i$ depends on the way the judgement $\tdts{\y{r}}{C}{i}$ has been derived. So different derivations could lead to different values, making this definition incorrect. The value $\val{\y{rs}}_{\vecb{a} : C}^j$ seemingly depend on the chosen index $i$ used in $\val{\y{s}}_{\vecb{a} : C}^i$. A consequence of the main Theorem \ref{mainthm} is that the value $\val{\y{r}}_{\vecb{a} : C}^i$ is indeed independent (modulo $\comp$) of the derivation of the state judgement $\tdts{\y{r}}{C}{i}$. 

\begin{definition}
\label{factlimtermdef}
Given a pair $(\vecb{A},\vecb{a})$ with a $\Gamma$-environment $\vecb{A}$ and a $C$-environment $\vecb{a}$ such that $\vecb{A} \Pmap \vecb{a}$. The interpretation of a typed term $\tdt{\y{r}}{\Gamma}{\rho}$ with state judgement $\tdts{\y{r}}{C}{i}$ relative to $(\vecb{A},\vecb{a})$ is the pair $(\val{\y{r}}_{\vecb{A}},\val{\y{r}}_{\vecb{a} : C}^i) \in \val{\rho} \times \val{\rho}_i$.
\end{definition}

One would expect that $\vecb{A} \Pmap \vecb{a}$ implies $\val{\y{r}}_{\vecb{A}} \Pmap \val{\y{r}}_{\vecb{a} : C}^i$. Indeed, this is the content of the next theorem, which is also the main theorem.

\begin{theorem}
\label{mainthm}
Given a typed term $\tdt{\y{r}}{\Gamma}{\rho}$ with $\tdts{\y{r}}{C}{i}$. If $\vecb{A} \in \val{\Gamma}$, $\vecb{a} \in \val{\Gamma}_C$ are variable assignments with $\vecb{A} \Pmap \vecb{a}$, then 
\[
\val{\y{r}}_{\vecb{A}} \Pmap \val{\y{r}}_{\vecb{a} : C}^i.
\]
\end{theorem}

\begin{proof}
The proof proceeds by induction on the derivation of $\tdt{\y{r}}{C}{i}$. For a variable $\yx{k}$ of a positive type we apply Lemma \ref{embpmaplem}. This yields
\[
\val{\yx{k}}_{\vecb{A}} = A_k \Pmap \emb{i_k}{j}(a_k) = \val{\yx{k}}_{\vecb{a} : C}^j,
\]
since $A_k \Pmap a_k$. A similar consideration holds for negative types. For application and abstraction the claim follows from the fact that $\Pmap$ is a logical relation.
\end{proof}

Consider the case of a closed term $\y{r} : \rho$. Let us assume that $\rho$ is some data type, then the theorem states that $\val{\y{r}} \Pmap \val{\y{r}}^i$. Recall that the predecessor relation $\pmap$ is seen as a directed equality between two states of the same object. We can thus read $\Pmap$, which is nothing more than $\pmap$ applied limits, as an equality between $\val{\y{r}}$ and $\val{\y{r}}^i$. As indicated in Section \ref{extsec}, $\val{\y{r}}^i$ is a strong form of an approximation of $\val{\y{r}}$. 

Due to the use of logical relations on higher types, application and $\lambda$-abstraction respect this equality. Since the type $\prop$ of propositions is part of the calculus, the equality holds also for truth values. The relation $\val{\y{r}} \Pmap \val{\y{r}}^i$ for type $\prop$, with the only values $true$ and $false$ in classical logic, is the identity. Consequently, truth in the limit and truth in a sufficiently large stage of the factor system coincide. Theorem \ref{mainthm} is the main result, and at the same time this theorem is necessary to show that the basic definition of the interpretation of a term, as introduced in Definition \ref{facttermdef}, is correct.

\begin{corollary}
\label{appcor1}
Let $\tdt{\y{r}}{\Gamma}{\rho}$ be a typed term with $i \leq i'$, $\tdts{\y{r}}{C}{i}$ and $\tdts{\y{r}}{C}{i'}$. Then $\val{\y{r}}_{\vecb{a} : C}^{i'} \pmap \val{\y{r}}_{\vecb{a} : C}^i$ holds for all $C$-environments $\vecb{a} \in \val{\Gamma}_C$.
\end{corollary}

\begin{proof}
Define $\vecb{A} := Emb_{C}(\vecb{a})$. Then $\vecb{A} \Pmap \vecb{a}$ follows from Condition (\ref{embcond}), hence $\val{\y{r}}_{\vecb{A}} \Pmap \val{\y{r}}_{\vecb{a} : C}^{i'}$ and $\val{\y{r}}_{\vecb{A}} \Pmap \val{\y{r}}_{\vecb{a} : C}^i$ by Theorem \ref{mainthm}, and therefore $\val{\y{r}}_{\vecb{a} : C}^{i'} \pmap \val{\y{r}}_{\vecb{a} : C}^i$ by Condition (\ref{pmapeq}).
\end{proof}

The independence of the value $\val{\y{r}}_{\vecb{a} : C}^i$ (modulo $\comp$) from the state judgement follows from Corollary \ref{appcor1} if we take $i = i'$.

\section{Conclusion and Further Work}

We presented a model that can be used to interpret a fragment of the simply typed $\lambda$-calculus (which we called \emph{core fragment}) based on the assumption that infinite sets are potential infinite. We gave a formalization of the potential infinite based on the filters $\fil$ and sets $\ll$ within these filters. This allows one to avoid any notion of actual infinity. The function space in this model is a family of finite function spaces and the model already has an interpretation of the logical type $\prop$ with the usual two truth values. 

We introduced an interpretation of $\lambda$-terms from the core fragment which has a dynamic part, the factor system, and a static part, the limit of the factor system. Nevertheless, both parts are necessary, because from a dynamic point of view, the construction of a limit is not the end of the process, but an intermediate state. 

The next step is to extend the $\lambda$-calculus to constants, in particular to the logical constants of implication and universal quantifier. The challenge here is that the universal quantifier is not continuous and cannot be interpreted as a higher-order functional. The solution is to introduce an additional rule with a side condition $C \ll i$ on the state judgments. The correctness of the interpretation uses the fact that propositions have ``stable truth values'' in the sense that for each proposition there is a stage in the model where the truth value of the proposition does not change anymore during further extensions. A corresponding interpretation for first-order logic has been given in \cite{eberl2022RML}.

The reflection principle from Theorem \ref{mainthm} then states that all objects and propositions (propositions are specific terms of type $\prop$) in the limit are reflected in a sufficiently large state in the system. The interpretation is possible on the core fragment of simple type theory, not for all terms. However, this fragment includes a version of a classical higher-order logic. Unlike other models, such as domain theory \cite{abramsky1994domain}, it does not require (actual) infinite sets at all, and it includes logic.\footnote{The idea of a logic-enriched type theory \cite{gambino2006generalised} separates propositions from the underlying type theoretic framework. In our approach the logic is inside type theory, but the interpretation of the universal quantifier has an extra treatment.}

We do not expect that the theorem about the reflection principle can be extended from the core fragment to all terms of simply typed $\lambda$-calculus without some kind of further restrictions. This is because variables of function type correspond to arbitrary functions, which can also be used to define higher-order functions by $\lambda$-abstraction. From the perspective of extensibility, this necessitates the consideration of both covariant (with respect to the codomain) and contravariant (with respect to the domain) extensions simultaneously. For instance, the totality of higher-order functions requires that, for each argument, which could be a function $f$ of type $nat \to nat$, there must be a value. However, $f$ is not given as a single entity. So one must identify a property that guarantees totality and conditions that ensure invariance, that is, existing properties must be preserved by future extensions of $f$. This is a significant challenge for higher-order functions and it is related to properties of the filter $\fil$. It is also the reason that for arbitrary functions, there are no simple state judgements as formulated in the rules (\textsc{Var$+$}) and (\textsc{Var$-$}). A future task will be to identify conditions under which the core fragment can be extended to include function variables.

\bibliography{ReflPotInf}

\end{document}